\documentclass[11pt]{article}
\usepackage{amssymb,amsthm}
\theoremstyle{plain}
\newtheorem{thm}{Theorem}[section]
\newtheorem{lem}[thm]{Lemma}
\newtheorem{prop}[thm]{Proposition}
\theoremstyle{definition}
\newtheorem{example}{Example}[section]
\theoremstyle{remark}
\newtheorem*{rem}{Remark}
\newcommand{\Gal}{\mathrm{Gal}}
\newcommand{\ord}{\mathrm{ord}}
\newcommand{\QQ}{\mathbb Q}
\newcommand{\FF}{\mathbb F}
\newcommand{\KK}{\mathbb K}
\newcommand{\EE}{\mathbb E}
\newcommand{\ZZ}{\mathbb Z}
\newcommand{\LL}{\mathbb L}
\newcommand{\CO}{\mathcal O}
\newcommand{\FP}{\mathfrak P}
\newcommand{\Fp}{\mathfrak p}
\newcommand{\bi}{\mathbf i}
\begin{document}
\title{On the Generality of \mbox{$1+\bi$} as a Non-Norm Element\thanks{This work was supported by the Stanford Graduate Fellowship. The material in this paper was submitted in part for presentations at the 2010 IEEE International Symposium on Information Theory.}}
\author{Hua-Chieh Li\thanks{Department of Mathematics, National Taiwan Normal University, Taipei 116, Taiwan. E-mail: li@math.ntnu.edu.tw.}
\and
Ming-Yang Chen\thanks{Department of Electrical Engineering, Stanford University, Stanford, CA 94305, U.S.A. E-mail: chenmy@stanford.edu.}
\and
John M. Cioffi\thanks{Department of Electrical Engineering, Stanford University, Stanford, CA 94305, U.S.A. E-mail: cioffi@stanford.edu.}}
\maketitle
\begin{abstract}
  Full-rate space-time block codes with nonvanishing determinants have been extensively designed with cyclic division algebras. For these designs, smaller pairwise error probabilities of maximum likelihood detections require larger normalized diversity products, which can be obtained by choosing integer non-norm elements with smaller absolute values. All known methods have constructed \mbox{$1+\bi$} and \mbox{$2+\bi$} to be integer non-norm elements with the smallest absolute values over QAM for the number of transmit antennas \mbox{$n$}: \mbox{$\{n:5\leq n\leq 40,8\nmid n\}$} and \mbox{$\{n:5\leq n\leq 40,8\mid n\}$}, respectively. Via explicit constructions, this paper proves that \mbox{$1+\bi$} is an integer non-norm element with the smallest absolute value over QAM for every \mbox{$n\geq 5$}.
\end{abstract}
\section{Introduction}
  A sufficient condition, called {\em nonvanishing determinants}, is shown in~\cite{EKPKL06} for a full-rate (\mbox{$n^2$} input symbols in an \mbox{$n\times n$} transmission matrix) space-time block code achieving the optimal diversity-multiplexing gains tradeoff~\cite{ZT03}. Full-rate space-time block codes with nonvanishing determinants have been extensively designed with cyclic division algebras, e.g.,~\cite{EKPKL06,GX09,KR05,LC09}. For these designs, smaller pairwise error probabilities of maximum likelihood detections require larger normalized diversity products, which can be obtained by choosing integer non-norm elements with smaller absolute values~\cite{GX09,LC09}. All known methods have constructed \mbox{$1+\bi$} and \mbox{$2+\bi$} to be integer non-norm elements with the smallest absolute values over quadrature amplitude modulations (QAM) for the number of transmit antennas \mbox{$n$}: \mbox{$\{n:5\leq n\leq 40,8\nmid n\}$} and \mbox{$\{n:5\leq n\leq 40,8\mid n\}$}, respectively~\cite{LC09}.

  Via explicit constructions, this paper proves that \mbox{$1+\bi$} is an integer non-norm element with the smallest absolute value over QAM for every \mbox{$n\geq 5$}. Section~\ref{prelim} introduces some fundamental definitions in algebraic number theory and discusses their properties, which are helpful in deriving the new constructional procedure in Section~\ref{method}. Section~\ref{numerical} confirms numerically the improvement in normalized diversity products by adopting \mbox{$1+\bi$} as a non-norm element instead of \mbox{$2+\bi$}. Finally, conclusions are drawn in Section~\ref{conc}.

  Throughout this article, \mbox{$\QQ$} and \mbox{$\ZZ$} mean the field consisting of all rational numbers and the ring consisting of all integers, respectively. For a field \mbox{$\EE$}, the sets of all algebraic integers and nonzero elements therein are denoted by \mbox{$\CO_\EE$} and \mbox{$\EE^*$}, respectively. Moreover, \mbox{$\EE_\Fp$} represents the completion of \mbox{$\EE$} with valuation corresponding to a nonzero prime ideal \mbox{$\Fp$} of \mbox{$\CO_\EE$}. \mbox{$\zeta_m$} is a primitive \mbox{$m$-th} root of unity.
\section{Preliminary Knowledge in Algebraic Number Theory}
\label{prelim}
  The following two paragraphs briefly mention some useful tools in ramification theory~\cite[Ch. I, Sec. 6]{J96}.

  Let \mbox{$\EE$} be a number field and \mbox{$\FF$} be an abelian extension over \mbox{$\EE$} with degree \mbox{$n$}. A nonzero prime ideal \mbox{$\FP$} of \mbox{$\CO_\FF$} is said to {\em lie over} another nonzero prime ideal \mbox{$\Fp$} of \mbox{$\CO_\EE$}, written as \mbox{$\FP\mid \Fp$}, if \mbox{$\FP\cap \CO_\EE=\Fp$}. Each \mbox{$\Fp\CO_\FF$} with \mbox{$\Fp$} a nonzero prime ideal of \mbox{$\CO_\EE$} has the unique (up to a reindexing) factorization in \mbox{$\CO_\FF$}: \mbox{$\Fp\CO_\FF=(\FP_1\cdots\FP_g)^{e}$} where \mbox{$\FP_1,\ldots,\FP_g$} are distinct nonzero prime ideals of \mbox{$\CO_\FF$} with \mbox{$\FP_1,\ldots,\FP_g\mid\Fp$}, and \mbox{$e$} is a positive integer known as the {\em ramification index} of each \mbox{$\FP_1,\ldots,\FP_g$} over \mbox{$\Fp$}. Since \mbox{$\CO_\EE$} and \mbox{$\CO_\FF$} are both Dedekind domains, \mbox{$\Fp$} is a maximal ideal in \mbox{$\CO_\EE$} and likewise for those \mbox{$\FP_1,\ldots,\FP_g$} in \mbox{$\CO_\FF$}, implying that \mbox{$\CO_\EE/\Fp$} and \mbox{$\CO_\FF/\FP_1,\ldots,\CO_\FF/\FP_g$} are all fields. In fact, each \mbox{$\CO_\FF/\FP_i$} is a field extension of \mbox{$\CO_\EE/\Fp$} with the same degree \mbox{$f=[\CO_\FF/\FP_i:\CO_\EE/\Fp]$}, called the {\em residue class degree} of \mbox{$\FP_i$} over \mbox{$\Fp$}, such that
\begin{equation}
\label{sum}
  efg=n.
\end{equation}
  Of particular interests are the two extreme cases: \mbox{$f=n$} or \mbox{$e=n$}. In either one,~(\ref{sum}) automatically acknowledges that there is exactly one nonzero prime ideal of \mbox{$\CO_\FF$} lying over \mbox{$\Fp$}, i.e., \mbox{$g=1$}. If \mbox{$f=n$} then \mbox{$\Fp$} is said to be {\em inert} in \mbox{$\FF/\EE$}; if \mbox{$e=n$} then \mbox{$\Fp$} is said to be {\em totally ramified} in \mbox{$\FF/\EE$}.

  Suppose further that \mbox{$\LL/\EE$} is another abelian extension where \mbox{$\FF/\EE$} is a sub-extension. Let \mbox{$\Fp$}, \mbox{$\FP$}, and \mbox{$\wp$} be nonzero prime ideals of \mbox{$\CO_\EE$}, \mbox{$\CO_\FF$}, and \mbox{$\CO_\LL$}, respectively, with \mbox{$\FP\mid \Fp$} and \mbox{$\wp\mid \FP$}. The respective ramification indices \mbox{$e$}, \mbox{$e'$}, and \mbox{$e''$} of \mbox{$\wp$} over \mbox{$\Fp$}, \mbox{$\wp$} over \mbox{$\FP$}, and \mbox{$\FP$} over \mbox{$\Fp$}, and the respective residue class degrees \mbox{$f$}, \mbox{$f'$}, and \mbox{$f''$} of \mbox{$\wp$} over \mbox{$\Fp$}, \mbox{$\wp$} over \mbox{$\FP$}, and \mbox{$\FP$} over \mbox{$\Fp$} satisfy
\begin{equation}
\label{tower}
  e=e'e'' \mbox{ and } f=f'f''.
\end{equation}
  Together~(\ref{sum}) with~(\ref{tower}), if a nonzero prime ideal \mbox{$\Fp$} of \mbox{$\CO_\EE$} is inert (respectively, totally ramified) in \mbox{$\LL/\EE$}, then it is also inert (respectively, totally ramified) in every sub-extension \mbox{$\FF/\EE$} of \mbox{$\LL/\EE$}.

  Let's focus on the situation where \mbox{$\FF/\EE$} is a cyclic extension. Recall that the global norm \mbox{$N_{\FF/\EE}(\gamma)$} of \mbox{$\gamma\in\FF$} is $$N_{\FF/\EE}(\gamma)=\prod_{\sigma\in\Gal(\FF/\EE)}\sigma(\gamma)\in \EE.$$ For every \mbox{$\gamma\in\EE^*$}, \mbox{$N_{\FF/\EE}(\gamma)=\gamma^n$} implies \mbox{$\gamma^n\in N_{\FF/\EE}(\FF^*)$}. Hence, the order of \mbox{$\gamma$} modulo \mbox{$N_{\FF/\EE}(\FF^*)$} always divides \mbox{$n$}. An element \mbox{$\gamma\in\EE^*$} is called a {\em non-norm element} of \mbox{$\FF/\EE$} if the order of \mbox{$\gamma$} modulo \mbox{$N_{\FF/\EE}(\FF^*)$} is \mbox{$n$}. For a pair of nonzero prime ideals \mbox{$\Fp$} of \mbox{$\CO_\EE$} and \mbox{$\FP$} of \mbox{$\CO_\FF$} with \mbox{$\FP\mid\Fp$}, \mbox{$\FF_\FP/\EE_\Fp$} is a Galois extension with \mbox{$\Gal(\FF_\FP/\EE_\Fp)$} isomorphic to a subgroup of \mbox{$\Gal(\FF/\EE)$}~\cite[Ch. III, Thm. 1.2]{J96}. Thus \mbox{$\FF_\FP/\EE_\Fp$} is a cyclic extension, too. The local norm \mbox{$N_{\FF_\FP/\EE_\Fp}(\gamma)$} of \mbox{$\gamma\in\FF\subseteq\FF_\FP$} is $$N_{\FF_\FP/\EE_\Fp}(\gamma)= \prod_{\sigma\in\Gal(\FF_\FP/\EE_\Fp)}\sigma(\gamma)\in \EE_\Fp.$$ Lemma~\ref{hasse} connects global and local norms.
\begin{lem}[{\cite[Ch. V, Thm. 4.6]{J96}}]
\label{hasse}
  For each \mbox{$\gamma \in\EE^*$}, if \mbox{$\gamma\not\in N_{\FF_\FP/\EE_\Fp}(\FF_\FP^*)$} for some pair of nonzero prime ideals \mbox{$\Fp$} of \mbox{$\CO_\EE$} and \mbox{$\FP$} of \mbox{$\CO_\FF$} with \mbox{$\FP\mid \Fp$}, then \mbox{$\gamma\not\in N_{\FF/\EE}(\FF^*)$}.
\end{lem}
  Lemma~\ref{comp} shows that it suffices to consider those cyclic extensions with degrees equal to prime powers and a pre-specified non-norm element.
\begin{lem}
\label{comp}
  Given two cyclic extensions \mbox{$\FF/\EE$} and \mbox{$\KK/\EE$} with degrees \mbox{$n_1$} and \mbox{$n_2$}, respectively, if \mbox{$\gcd(n_1,n_2)=1$} and \mbox{$\gamma\in \EE^*$} is a non-norm element of both extensions, then \mbox{$\FF\KK/\EE$} is a cyclic extension with degree \mbox{$n_1n_2$} and \mbox{$\gamma$} a non-norm element.
\end{lem}
\begin{proof}
  From \mbox{$\gcd(n_1,n_2)=1$}, \mbox{$\FF\KK/\EE$} is a cyclic extension with degree \mbox{$n_1n_2$}. Write \mbox{$\LL=\FF\KK$}. Suppose, on the contrary, that \mbox{$\gamma$} is not a non-norm element of \mbox{$\LL/\EE$}; i.e., the order of \mbox{$\gamma$} modulo \mbox{$N_{\LL/\EE}(\LL^*)$} is a proper divisor of \mbox{$n_1n_2$}. Then there exists a prime number \mbox{$q\mid n_1n_2$} such that $$\gamma^{n_1n_2/q}\in N_{\LL/\EE}(\LL^*).$$ Assume that \mbox{$q\mid n_1$} without loss of generality. By the transitivity of norm~\cite[Ch. I, Cor. 5.4]{J96}, $$\gamma^{n_1n_2/q}\in N_{\FF/\EE}(N_{\LL/\FF}(\LL^*))\subseteq N_{\FF/\EE}(\FF^*).$$ In other words, the order of \mbox{$\gamma$} modulo \mbox{$N_{\FF/\EE}(\FF^*)$} has to divide \mbox{$n_1n_2/q$}. It also divides \mbox{$[\FF:\EE]=n_1$}. The order of \mbox{$\gamma$} modulo \mbox{$N_{\FF/\EE}(\FF^*)$} must be a divisor of \mbox{$\gcd(n_1n_2/q,n_1)=n_1/q$}, a contradiction to \mbox{$\gamma$} being a non-norm element of \mbox{$\FF/\EE$}.
\end{proof}

  If a nonzero prime ideal \mbox{$\Fp$} of \mbox{$\CO_\EE$} is totally ramified in \mbox{$\FF/\EE$}, and the extension degree \mbox{$[\FF:\EE]$} is not divisible by the characteristic of field \mbox{$\CO_\EE/\Fp$}, then \mbox{$\Fp$} is said to be {\em totally and tamely ramified} in \mbox{$\FF/\EE$}. Some later serviceable instances are illustrated in Example~\ref{ttr}.
\begin{example}
\label{ttr}
  Let \mbox{$p$} be a prime number. Then \mbox{$[\QQ(\zeta_p):\QQ]=p-1$} and \mbox{$p$} is totally ramified in \mbox{$\QQ(\zeta_p)/\QQ$}~\cite[Ch. I, Thm. 10.1]{J96}. Since the characteristic of \mbox{$\ZZ/p\ZZ$} is \mbox{$p$}, \mbox{$p$} is totally and tamely ramified in \mbox{$\QQ(\zeta_p)/\QQ$}.

  If further \mbox{$p\equiv 1\pmod{4}$} then there exist two positive integers \mbox{$a$} and \mbox{$b$} such that \mbox{$p=a^2+b^2$}, leading to exactly two distinct nonzero prime ideals of \mbox{$\CO_{\QQ(\bi)}=\ZZ[\bi]$} lying over \mbox{$p$}: \mbox{$(a+b\bi)\ZZ[\bi]$} and \mbox{$(a-b\bi)\ZZ[\bi]$}. Let \mbox{$\Fp$} be either of them. By~(\ref{sum}), \mbox{$\Fp$} has ramification index \mbox{$1$} over \mbox{$p$}. Fix a nonzero prime ideal \mbox{$\wp$} of \mbox{$\CO_{\QQ(\zeta_p,\bi)}$} with \mbox{$\wp\mid \Fp$} and let \mbox{$\FP=\wp\cap\CO_{\QQ(\zeta_p)}$} accordingly. Then \mbox{$\FP\mid p$}. Since \mbox{$p$} is totally ramified in \mbox{$\QQ(\zeta_p)/\QQ$}, \mbox{$\FP$} has ramification index \mbox{$p-1$} over \mbox{$p$}. By~(\ref{tower}), the tower of field extensions \mbox{$\QQ\subset \QQ(\zeta_p)\subset \QQ(\zeta_p,\bi)$} gives that the ramification index of \mbox{$\wp$} over \mbox{$p$} is at least \mbox{$p-1$}. Since \mbox{$\Fp$} has ramification index \mbox{$1$} over \mbox{$p$}, again by~(\ref{tower}) the tower of field extensions \mbox{$\QQ\subset\QQ(\bi)\subset\QQ(\zeta_p,\bi)$} yields that the ramification index of \mbox{$\wp$} over \mbox{$\Fp$} is at least \mbox{$p-1$}. From~(\ref{sum}) and \mbox{$[\QQ(\zeta_p,\bi):\QQ(\bi)]=p-1$}, the ramification index of \mbox{$\wp$} over \mbox{$\Fp$} must be \mbox{$p-1$}; i.e., \mbox{$\Fp$} is totally ramified in \mbox{$\QQ(\zeta_p,\bi)/\QQ(\bi)$}. Moreover, \mbox{$\ZZ[\bi]/\Fp$} is a field extension of \mbox{$\ZZ/p\ZZ$}; thus the characteristic of \mbox{$\ZZ[\bi]/\Fp$} must be \mbox{$p$}. \mbox{$\Fp$} is totally and tamely ramified in \mbox{$\QQ(\zeta_p,\bi)/\QQ(\bi)$}.
\end{example}
  Theorem~\ref{locnon} presents our new sufficient condition for obtaining non-norm elements.
\begin{thm}
\label{locnon}
  Let \mbox{$q$} be a prime number, \mbox{$k$} be a positive integer, \mbox{$\FF/\EE$} be a cyclic extension with degree \mbox{$q^k$}, and \mbox{$\Fp$} be totally and tamely ramified in \mbox{$\FF/\EE$}. If \mbox{$\gamma\in\CO_\EE\setminus\Fp$} and \mbox{$x^q\equiv \gamma\pmod{\Fp}$} is not solvable in \mbox{$\CO_\EE$}, then \mbox{$\gamma$} is a non-norm element of \mbox{$\FF/\EE$}.
\end{thm}
\begin{proof}
  Since \mbox{$\Fp$} is totally ramified in \mbox{$\FF/\EE$}, there is a unique nonzero prime ideal \mbox{$\FP$} of \mbox{$\CO_\FF$} with \mbox{$\FP\mid\Fp$}. By Lemma~\ref{hasse}, to show that \mbox{$\gamma$} is a non-norm element of \mbox{$\FF/\EE$}, it suffices to prove $$\gamma^{q^{k-1}}\not\in N_{\FF_\FP/\EE_\Fp}(\FF_\FP^*).$$ The condition \mbox{$\gamma\in\CO_\EE\setminus\Fp$} gives $$\gamma^{q^{k-1}}\in\CO_{\EE_\Fp}\setminus\Fp\CO_{\EE_\Fp};$$ the condition that \mbox{$\Fp$} is totally and tamely ramified in \mbox{$\FF/\EE$} implies that \mbox{$\FF_\FP/\EE_\Fp$} is a totally and tamely ramified extension with degree \mbox{$q^k$}~\cite[Ch. II, Thm. 3.8]{J96}. Thereby every \mbox{$\alpha\in\CO_{\EE_\Fp}\setminus\Fp\CO_{\EE_\Fp}$} satisfies \mbox{$\alpha\in N_{\FF_\FP/\EE_\Fp}(\FF_\FP^*)$} if and only if $$\alpha\equiv \beta^{q^k}\pmod{\Fp\CO_{\EE_\Fp}}$$ for some \mbox{$\beta\in\CO_{\EE_\Fp}$}~\cite[Ch. IV, Sec. 1.5]{FV02}. Together with \mbox{$\CO_{\EE_\Fp}/\Fp\CO_{\EE_\Fp}\simeq\CO_{\EE}/\Fp$}~\cite[Ch. II, Cor. 2.7]{J96}, it leaves to show that
\begin{equation}
\label{coneq}
  x^{q^k}\equiv \gamma^{q^{k-1}}\pmod{\Fp}
\end{equation}
  is not solvable in \mbox{$\CO_\EE$}. Suppose, on the contrary, that~(\ref{coneq}) has a solution \mbox{$\theta\in \CO_\EE$}. Since \mbox{$\CO_\EE/\Fp$} is a finite field, \mbox{$(\CO_\EE/\Fp)^*$} is a cyclic group. Let \mbox{$\lambda\in\CO_\EE\setminus \Fp$} be a generator of \mbox{$(\CO_\EE/\Fp)^*$}. Then \mbox{$\theta\equiv\lambda^s\pmod{\Fp}$} and \mbox{$\gamma\equiv\lambda^t\pmod{\Fp}$} for some positive integers \mbox{$s$} and \mbox{$t$} such that $$\lambda^{s q^k}\equiv\lambda^{tq^{k-1}}\pmod{\Fp};$$ i.e., the order of \mbox{$(\CO_\EE/\Fp)^*$} divides \mbox{$q^{k-1}(sq-t)$}. On the other hand, \mbox{$q^k$} must divide the order of \mbox{$(\CO_\EE/\Fp)^*$} because \mbox{$\Fp$} is totally and tamely ramified in \mbox{$\FF/\EE$}~\cite[Ch. IV, Sec. 1.5]{FV02}. Therefore \mbox{$q\mid t$} and \mbox{$\lambda^{t/q}$} becomes a solution of \mbox{$x^q\equiv \gamma\pmod{\Fp}$} in \mbox{$\CO_\EE$}, a contradiction.
\end{proof}
\section{The Generality of \mbox{$1+\bi$} as a Non-Norm Element}
\label{method}
  This section first shows that for every odd \mbox{$n$}, a cyclic extension over \mbox{$\QQ(\bi)$} with degree \mbox{$n$} and \mbox{$1+\bi$} a non-norm element can always be constructed. This paper aims to establish such an extension for each positive integer \mbox{$n$}. According to Lemma~\ref{comp}, it suffices to further construct a cyclic extension over \mbox{$\QQ(\bi)$} with degree \mbox{$2^k$} and \mbox{$1+\bi$} a non-norm element for each positive integer \mbox{$k$}.

  Lemma~\ref{spl} has been used extensively in~\cite{EKPKL06,GX09,KR05,LC09}.
\begin{lem}[{\cite[Thm. 1]{KR05}}]
\label{spl}
  Let \mbox{$\FF/\QQ$} be a cyclic extension with degree \mbox{$n$} and \mbox{$\FF\cap\QQ(\bi)=\QQ$}, \mbox{$\ell$} be an inert prime number in \mbox{$\FF/\QQ$}, and \mbox{$\Fp$} be a nonzero prime ideal of \mbox{$\ZZ[\bi]$} with \mbox{$\Fp\mid \ell$} and residue class degree $1$ over $\ell$. If \mbox{$\gamma\in\ZZ[\bi]$} satisfies \mbox{$\gamma\in\Fp\setminus\Fp^2$}, then \mbox{$\gamma$} is a non-norm element of \mbox{$\FF(\bi)/\QQ(\bi)$}.
\end{lem}
  Denote \mbox{$\varphi(\cdot)$} as Euler's totient function. If \mbox{$m$} is an odd prime power then \mbox{$\QQ(\zeta_m)/\QQ$} is a cyclic extension with degree \mbox{$\varphi(m)$}. In that case, there is a unique intermediate field \mbox{$\FF$} of \mbox{$\QQ(\zeta_{m})/\QQ$} such that \mbox{$\FF/\QQ$} is a cyclic extension with degree \mbox{$n$} whenever \mbox{$n\mid\varphi(m)$}. Motivated by Lemma~\ref{spl} and \mbox{$\varphi(q^{k+1})=q^k(q-1)$}, we naturally consider \mbox{$m=q^{k+1}$} and \mbox{$n=q^k$} for each odd prime number \mbox{$q$} and positive integer \mbox{$k$}, shown as follows.
\begin{prop}
\label{power}
  Let \mbox{$q$} be an odd prime number, \mbox{$k$} be a positive integer, and \mbox{$\FF/\QQ$} be the unique \mbox{degree-$q^k$} sub-extension of \mbox{$\QQ(\zeta_{q^{k+1}})/\QQ$}. If \mbox{$\ell$} is a prime number such that \mbox{$\ell\neq q$} and \mbox{$q^2\nmid \ell^{q-1}-1$} then \mbox{$\ell$} is inert in \mbox{$\FF/\QQ$}.
\end{prop}
\begin{proof}
  For a nonzero prime ideal \mbox{$\wp$} of \mbox{$\CO_{\QQ(\zeta_{q^{k+1}})}$} with \mbox{$\wp\mid\ell$}, the residue class degree \mbox{$f$} of \mbox{$\wp$} over \mbox{$\ell$} is the smallest positive integer \mbox{$r$} such that \mbox{$\ell^r\equiv 1\pmod{q^{k+1}}$}, denoted by \mbox{$\ord_{q^{k+1}}(\ell)$}~\cite[Ch. 13, Sec. 2, Thm. 2]{IR90}. Fermat's little theorem and \mbox{$\ell\ne q$} imply \mbox{$q\mid\ell^{q-1}-1$}; thus \mbox{$q^2\nmid \ell^{q-1}-1$} results in \mbox{$\ell^{q-1}=1+sq$} with \mbox{$q\nmid s$}. From \mbox{$\ell^{q(q-1)}=(1+sq)^q=1+sq^2+tq^3$}, \mbox{$q^2\mid\ell^{q(q-1)}-1$} and \mbox{$q^3\nmid \ell^{q(q-1)}-1$}. Continuing this argument leads to $$q^{k+1}\nmid \ell^{q^{k-1}(q-1)}-1,$$ i.e., \mbox{$\ord_{q^{k+1}}(\ell)\nmid q^{k-1}(q-1)$}. Moreover, Euler's theorem states $$\ord_{q^{k+1}}(\ell)\mid\varphi(q^{k+1})=q^k(q-1).$$ Therefore \mbox{$q^k\mid\ord_{q^{k+1}}(\ell)=f$}. Let \mbox{$\FP=\wp\cap\CO_\FF$}, and \mbox{$f'$} and \mbox{$f''$} be the residue class degrees of \mbox{$\wp$} over \mbox{$\FP$} and \mbox{$\FP$} over \mbox{$\ell$}, respectively. By~(\ref{tower}), \mbox{$q^k\mid f'f''$}, and by~(\ref{sum}), $$f'\mid [\QQ(\zeta_{q^{k+1}}):\FF]=q^k(q-1)/q^k=q-1.$$ Hence, \mbox{$q^k\mid f''$}. Again, by~(\ref{sum}) and \mbox{$[\FF:\QQ]=q^k$}, \mbox{$f''=[\FF:\QQ]$}; \mbox{$\ell$} is inert in \mbox{$\FF/\QQ$}.
\end{proof}
  Together Lemma~\ref{spl} with Proposition~\ref{power}, for each odd prime number \mbox{$q$} with \mbox{$q^2\nmid 2^{q-1}-1$} and positive integer \mbox{$k$}, a cyclic extension over \mbox{$\QQ(\bi)$} with degree \mbox{$q^k$} and \mbox{$1+\bi$} a non-norm element can always be constructed. Nonetheless, there do exist odd prime numbers \mbox{$q$} with \mbox{$q^2\mid 2^{q-1}-1$}, called the Wieferich primes. Although the only known Wieferich primes are \mbox{$q=1093$} and \mbox{$3511$}, Proposition~\ref{residue} below is capable of handling more general scenarios.
\begin{prop}
\label{residue}
  Let \mbox{$q$} be a prime number, \mbox{$k$} be a positive integer, \mbox{$p$} be an odd prime number with \mbox{$q^k\mid p-1$}, and \mbox{$\FF/\QQ(\bi)$} be the unique \mbox{degree-$q^k$} sub-extension of \mbox{$\QQ(\zeta_{p},\bi)/\QQ(\bi)$}. If \mbox{$x^q\equiv 2\pmod{p}$} is not solvable in \mbox{$\ZZ$}, then \mbox{$1+\bi$} is a non-norm element of \mbox{$\FF/\QQ(\bi)$}.
\end{prop}
\begin{proof}
  Let \mbox{$\KK/\QQ$} be the unique \mbox{degree-$q^k$} sub-extension of \mbox{$\QQ(\zeta_{p})/\QQ$}. From Example~\ref{ttr}, \mbox{$p$} is totally and tamely ramified in \mbox{$\KK/\QQ$}. Theorem~\ref{locnon} states that \mbox{$2$} is a non-norm element of \mbox{$\KK/\QQ$}. Suppose, on the contrary, that \mbox{$1+\bi$} is not a non-norm element of \mbox{$\FF/\QQ(\bi)$}, i.e., \mbox{$N_{\FF/\QQ(\bi)}(\alpha)=(1+\bi)^r$} for some \mbox{$\alpha\in\FF^*$} and \mbox{$1\le r<q^k$}. By the transitivity of norm, $$2^r=N_{\QQ(\bi)/\QQ}((1+\bi)^r)=N_{\QQ(\bi)/\QQ}(N_{\FF/\QQ(\bi)}(\alpha))=N_{\FF/\QQ}(\alpha)=N_{\KK/\QQ}(N_{\FF/\KK}(\alpha))\in N_{\KK/\QQ}(\KK^*),$$ a contradiction to \mbox{$2$} being a non-norm element of \mbox{$\KK/\QQ$}.
\end{proof}
  By Chebotarev's density theorem or~\cite[Thm. 4]{AR51}, there exist infinitely many odd prime numbers \mbox{$p$} such that \mbox{$q^k\mid p-1$} and \mbox{$x^q\equiv 2\pmod{p}$} is not solvable in \mbox{$\ZZ$} for each odd prime number \mbox{$q$} and positive integer \mbox{$k$}; thus a cyclic extension over \mbox{$\QQ(\bi)$} with degree \mbox{$q^k$} and \mbox{$1+\bi$} a non-norm element can always be constructed. Note that this statement is not true for \mbox{$q=2$} and \mbox{$k\ge 3$}; otherwise \mbox{$q^k\mid p-1$} implies \mbox{$8\mid p-1$} thereby that \mbox{$x^2\equiv 2\pmod{p}$} is always solvable in \mbox{$\ZZ$}. This is also the reason why there is no cyclic extension over \mbox{$\QQ$} with degree \mbox{$n$} such that \mbox{$2$} is inert whenever \mbox{$8\mid n$}. In other words, it is infeasible to find a cyclic extension over \mbox{$\QQ(\bi)$} with each of such degrees \mbox{$n$} and \mbox{$1+\bi$} a non-norm element by simply beginning at Lemma~\ref{spl}.

  Let \mbox{$n=2^k$} with \mbox{$k=1$} or \mbox{$2$}. Both \mbox{$x^2\equiv 2\pmod{3}$} and \mbox{$x^2\equiv 2\pmod{5}$} are not solvable in \mbox{$\ZZ$}. By Proposition~\ref{residue}, \mbox{$\QQ(\zeta_3,\bi)/\QQ(\bi)$} and \mbox{$\QQ(\zeta_5,\bi)/\QQ(\bi)$} are cyclic extensions with degrees \mbox{$2$} and \mbox{$4$}, respectively, and \mbox{$1+\bi$} a non-norm element. Integrated by Lemma~\ref{comp}, whenever \mbox{$8\nmid n$}, a cyclic extension over \mbox{$\QQ(\bi)$} with degree \mbox{$n$} and \mbox{$1+\bi$} a non-norm element can always be constructed.
\begin{example}[\mbox{$n=6$}]
  \mbox{$\QQ(\zeta_{9},\bi)/\QQ(\bi)$} is a cyclic extension with degree \mbox{$6$}. Since \mbox{$3^2\nmid 2^2-1$}, Lemma~\ref{spl} and Proposition~\ref{power} state that \mbox{$1+\bi$} is a non-norm element for the unique \mbox{degree-$3$} sub-extension \mbox{$\FF/\QQ(\bi)$} of \mbox{$\QQ(\zeta_{9},\bi)/\QQ(\bi)$}. Moreover, \mbox{$1+\bi$} is a non-norm element of \mbox{$\QQ(\zeta_3,\bi)/\QQ(\bi)$}. By Lemma~\ref{comp}, \mbox{$1+\bi$} is a non-norm element of the \mbox{degree-$6$} extension \mbox{$\FF(\zeta_{3})/\QQ(\bi)$}. In fact, \mbox{$\FF(\zeta_{3})=\QQ(\zeta_9,\bi)$} because both \mbox{$\FF$} and \mbox{$\QQ(\zeta_3,\bi)$} are subfields of \mbox{$\QQ(\zeta_9,\bi)$}.
\end{example}

  For \mbox{$n=2^k$} with \mbox{$k\ge 3$}, let \mbox{$p$} be an odd prime number with \mbox{$2^k\mid p-1$}, \mbox{$\Fp$} be a nonzero prime ideal of \mbox{$\ZZ[\bi]$} with \mbox{$\Fp\mid p$}, and \mbox{$\FF/\QQ(\bi)$} be the unique \mbox{degree-$2^k$} sub-extension of \mbox{$\QQ(\zeta_p,\bi)/\QQ(\bi)$}. Example~\ref{ttr} states that \mbox{$\Fp$} is totally and tamely ramified in \mbox{$\FF/\QQ(\bi)$}. From \mbox{$1+\bi\in\ZZ[\bi]\setminus\Fp$} and Theorem~\ref{locnon}, \mbox{$1+\bi$} is a non-norm element of \mbox{$\FF/\QQ(\bi)$} as long as \mbox{$x^2\equiv 1+\bi\pmod{\Fp}$} is not solvable in \mbox{$\ZZ[\bi]$}.

  By definitions, if \mbox{$1+\bi$} is a non-norm element of some cyclic extension \mbox{$\FF/\QQ(\bi)$} with degree \mbox{$2^r$}, then it is also a non-norm element of the unique \mbox{degree-$2^k$} sub-extension of \mbox{$\FF/\QQ(\bi)$} for every \mbox{$k\le r$}. The odd prime number \mbox{$1+2647\cdot 2^{1000}$} permits \mbox{$1+\bi$} to be a non-norm element of \mbox{$\QQ(\zeta_{1+2647\cdot 2^{1000}},\bi)/\QQ(\bi)$}. Indeed, \mbox{$1+\bi$} is always a generator of \mbox{$(\ZZ[\bi]/\Fp)^*$} whenever \mbox{$\Fp$} is a nonzero prime ideal of \mbox{$\ZZ[\bi]$} with \mbox{$\Fp\mid 1+2647\cdot 2^{1000}$}. Consequently, \mbox{$1+\bi$} is a non-norm element of the unique \mbox{degree-$2^k$} sub-extension of \mbox{$\QQ(\zeta_{1+2647\cdot 2^{1000}},\bi)/\QQ(\bi)$} for each \mbox{$k\le 1000$}. Generally speaking, Chebotarev's density theorem shows that for every \mbox{$k\ge 3$}, there exist infinitely many odd prime numbers \mbox{$p$} such that
\begin{enumerate}
\item \mbox{$p\equiv 1\pmod{2^k}$};
\item \mbox{$x^2\equiv 1+\bi\pmod{\Fp}$} is not solvable in
\mbox{$\ZZ[\bi]$} for a nonzero prime ideal \mbox{$\Fp$} of \mbox{$\ZZ[\bi]$} with \mbox{$\Fp\mid p$}.
\end{enumerate}
  Hence, a cyclic extension over \mbox{$\QQ(\bi)$} with degree \mbox{$2^k$} and \mbox{$1+\bi$} a non-norm element can always be constructed for each positive integer \mbox{$k$}.
\begin{thm}
\label{gen}
  There exists a cyclic extension over \mbox{$\QQ(\bi)$} with degree \mbox{$n$} and \mbox{$1+\bi$} a non-norm element for each positive integer \mbox{$n$}.
\end{thm}
  Recall~\cite[Lem. 3.1]{LC09} that for an odd prime number \mbox{$p$} with \mbox{$n\mid p-1$}, the unique \mbox{degree-$n$} sub-extension \mbox{$\FF/\QQ(\bi)$} of \mbox{$\QQ(\zeta_p,\bi)/\QQ(\bi)$} can be obtained through
\begin{equation}
\label{sub}
  \FF= \QQ(\eta,\bi)\mbox{ with }\eta=\sum_{i=0}^{\frac{p-1}{n}-1}\zeta_p^{c^{ni}},
\end{equation}
  where \mbox{$c$} is a primitive root modulo \mbox{$p$}. Example~\ref{real} demonstrates the generalized constructional procedure by instancing \mbox{$n=8$} and \mbox{$16$}.
\begin{example}[\mbox{$n=8$} and \mbox{$16$}]
\label{real}
  \mbox{$\QQ(\zeta_{17},\bi)/\QQ(\bi)$} is a cyclic extension with degree \mbox{$16$}. According to~(\ref{sub}), \mbox{$\eta=\zeta_{17}+\zeta_{17}^{-1}$} constructs \mbox{$\QQ(\eta,\bi)/\QQ(\bi)$} to be the unique \mbox{degree-$8$} sub-extension of \mbox{$\QQ(\zeta_{17},\bi)/\QQ(\bi)$}. From \mbox{$17\ZZ[\bi]=(1+4\bi)\ZZ[\bi](1-4\bi)\ZZ[\bi]$}, the residue class degree of \mbox{$1+4\bi$} over \mbox{$17$} is \mbox{$1$}; in other words, \mbox{$\ZZ[\bi]/(1+4\bi)\simeq\ZZ/17\ZZ$}. From \mbox{$\bi\equiv 4\pmod{1+4\bi}$}, to solve \mbox{$x^2\equiv 1+\bi\pmod{1+4\bi}$} in \mbox{$\ZZ[\bi]$} is equivalent to solve \mbox{$x^2\equiv 1+4\pmod{17}$} in \mbox{$\ZZ$}. Whereas \mbox{$x^2\equiv 5\pmod{17}$} is not solvable in \mbox{$\ZZ$}, \mbox{$1+\bi$} is a non-norm element of both the \mbox{degree-$8$} extension \mbox{$\QQ(\eta,\bi)/\QQ(\bi)$} and \mbox{degree-$16$} extension \mbox{$\QQ(\zeta_{17},\bi)/\QQ(\bi)$}.
\end{example}
  Table~\ref{table1} illustrates cyclic extensions over \mbox{$\QQ(\bi)$} with degrees ranging from \mbox{$2$} to \mbox{$100$} and \mbox{$1+\bi$} a non-norm element, by jointing those designated primitive roots of unity to \mbox{$\QQ(\bi)$}. Note that none of \mbox{$\{\pm 1,\pm\bi\}$} can be a non-norm element when the extension degree over \mbox{$\QQ(\bi)$}, i.e., the number of transmit antennas, \mbox{$n\ge 5$}; thus Theorem~\ref{gen} also proves that \mbox{$1+\bi$} is an integer non-norm element with the smallest absolute value over QAM for every \mbox{$n\geq 5$}.
\begin{rem}
  An analogous series of arguments can show how to construct a cyclic extension over \mbox{$\QQ(\zeta_3)$} with degree \mbox{$n$} and \mbox{$\sqrt{-3}$} a non-norm element for every positive integer $n$. Table~\ref{table2} illustrates cyclic extensions over \mbox{$\QQ(\zeta_3)$} with degrees ranging from \mbox{$2$} to \mbox{$100$} and \mbox{$\sqrt{-3}$} a non-norm element, by jointing those designated primitive roots of unity to \mbox{$\QQ(\zeta_3)$}. Note that none of \mbox{$\{\pm 1,\pm\zeta_3,\pm\zeta_3^2\}$} can be a non-norm element when \mbox{$n\ge 7$}. Similarly, \mbox{$\sqrt{-3}$} is an integer non-norm element with the smallest absolute value over hexagonal modulations (HEX) for every \mbox{$n\geq 7$}.
\end{rem}
\begin{table*}
\centering
\begin{tabular}{|c|c||c|c||c|c||c|c|}
\hline
  Extension & Roots & Extension & Roots & Extension & Roots & Extension & Roots \\
  Degrees & of Unity & Degrees & of Unity & Degrees & of Unity & Degrees & of Unity \\
\hline
  &  & \mbox{$26$} & \mbox{$\zeta_{53}$} & \mbox{$51$} & \mbox{$\zeta_{103}$} & \mbox{$76$} & \mbox{$\zeta_{229}$} \\
  \mbox{$2$} & \mbox{$\zeta_3$} & \mbox{$27$} & \mbox{$\zeta_{81}$} & \mbox{$52$} & \mbox{$\zeta_{53}$} & \mbox{$77$} & \mbox{$\zeta_{463}$} \\
  \mbox{$3$} & \mbox{$\zeta_7$} & \mbox{$28$} & \mbox{$\zeta_{29}$} & \mbox{$53$} & \mbox{$\zeta_{107}$} & \mbox{$78$} & \mbox{$\zeta_{169}$} \\
  \mbox{$4$} & \mbox{$\zeta_5$} & \mbox{$29$} & \mbox{$\zeta_{59}$} & \mbox{$54$} & \mbox{$\zeta_{81}$} & \mbox{$79$} & \mbox{$\zeta_{317}$} \\
  \mbox{$5$} & \mbox{$\zeta_{11}$} & \mbox{$30$} & \mbox{$\zeta_{61}$} & \mbox{$55$} & \mbox{$\zeta_{121}$} & \mbox{$80$} & \mbox{$\zeta_{187}$} \\
  \mbox{$6$} & \mbox{$\zeta_{9}$} & \mbox{$31$} & \mbox{$\zeta_{311}$} & \mbox{$56$} & \mbox{$\zeta_{493}$} & \mbox{$81$} & \mbox{$\zeta_{163}$} \\
  \mbox{$7$} & \mbox{$\zeta_{29}$} & \mbox{$32$} & \mbox{$\zeta_{97}$} & \mbox{$57$} & \mbox{$\zeta_{361}$} & \mbox{$82$} & \mbox{$\zeta_{83}$} \\
  \mbox{$8$} & \mbox{$\zeta_{17}$} & \mbox{$33$} &$\zeta_{67}$  & \mbox{$58$} & \mbox{$\zeta_{59}$} & \mbox{$83$} & \mbox{$\zeta_{167}$} \\
  \mbox{$9$} & \mbox{$\zeta_{19}$} & \mbox{$34$} & \mbox{$\zeta_{307}$} & \mbox{$59$} & \mbox{$\zeta_{709}$} & \mbox{$84$} & \mbox{$\zeta_{203}$} \\
  \mbox{$10$} & \mbox{$\zeta_{11}$} & \mbox{$35$} & \mbox{$\zeta_{71}$} & \mbox{$60$} & \mbox{$\zeta_{61}$} & \mbox{$85$} & \mbox{$\zeta_{1021}$} \\
  \mbox{$11$} & \mbox{$\zeta_{23}$} & \mbox{$36$} & \mbox{$\zeta_{37}$} & \mbox{$61$} & \mbox{$\zeta_{367}$} & \mbox{$86$} & \mbox{$\zeta_{173}$} \\
  \mbox{$12$} & \mbox{$\zeta_{13}$} & \mbox{$37$} & \mbox{$\zeta_{149}$} & \mbox{$62$} & \mbox{$\zeta_{373}$} & \mbox{$87$} & \mbox{$\zeta_{349}$} \\
  \mbox{$13$} & \mbox{$\zeta_{53}$} & \mbox{$38$} & \mbox{$\zeta_{361}$} & \mbox{$63$} & \mbox{$\zeta_{379}$} & \mbox{$88$} & \mbox{$\zeta_{391}$} \\
  \mbox{$14$} & \mbox{$\zeta_{29}$} & \mbox{$39$} & \mbox{$\zeta_{79}$} & \mbox{$64$} & \mbox{$\zeta_{193}$} & \mbox{$89$} & \mbox{$\zeta_{179}$} \\
  \mbox{$15$} & \mbox{$\zeta_{61}$} & \mbox{$40$} & \mbox{$\zeta_{187}$} & \mbox{$65$} & \mbox{$\zeta_{131}$} & \mbox{$90$} & \mbox{$\zeta_{181}$} \\
  \mbox{$16$} & \mbox{$\zeta_{17}$} & \mbox{$41$} & \mbox{$\zeta_{83}$} & \mbox{$66$} & \mbox{$\zeta_{67}$} & \mbox{$91$} & \mbox{$\zeta_{547}$} \\
  \mbox{$17$} & \mbox{$\zeta_{103}$} & \mbox{$42$} & \mbox{$\zeta_{147}$} & \mbox{$67$} & \mbox{$\zeta_{269}$} & \mbox{$92$} & \mbox{$\zeta_{235}$} \\
  \mbox{$18$} & \mbox{$\zeta_{19}$} & \mbox{$43$} & \mbox{$\zeta_{173}$} & \mbox{$68$} & \mbox{$\zeta_{515}$} & \mbox{$93$} & \mbox{$\zeta_{373}$} \\
  \mbox{$19$} & \mbox{$\zeta_{191}$} & \mbox{$44$} & \mbox{$\zeta_{115}$} & \mbox{$69$} & \mbox{$\zeta_{139}$} & \mbox{$94$} & \mbox{$\zeta_{283}$} \\
  \mbox{$20$} & \mbox{$\zeta_{25}$} & \mbox{$45$} & \mbox{$\zeta_{181}$} & \mbox{$70$} & \mbox{$\zeta_{211}$} & \mbox{$95$} & \mbox{$\zeta_{191}$} \\
  \mbox{$21$} & \mbox{$\zeta_{49}$} & \mbox{$46$} & \mbox{$\zeta_{139}$} & \mbox{$71$} & \mbox{$\zeta_{569}$} & \mbox{$96$} & \mbox{$\zeta_{679}$} \\
  \mbox{$22$} & \mbox{$\zeta_{67}$} & \mbox{$47$} & \mbox{$\zeta_{283}$} & \mbox{$72$} & \mbox{$\zeta_{323}$} & \mbox{$97$} & \mbox{$\zeta_{389}$} \\
  \mbox{$23$} & \mbox{$\zeta_{47}$} & \mbox{$48$} & \mbox{$\zeta_{119}$} & \mbox{$73$} & \mbox{$\zeta_{293}$} & \mbox{$98$} & \mbox{$\zeta_{197}$} \\
  \mbox{$24$} & \mbox{$\zeta_{119}$} & \mbox{$49$} & \mbox{$\zeta_{197}$} & \mbox{$74$} & \mbox{$\zeta_{149}$} & \mbox{$99$} & \mbox{$\zeta_{199}$} \\
  \mbox{$25$} & \mbox{$\zeta_{101}$} & \mbox{$50$} & \mbox{$\zeta_{101}$} & \mbox{$75$} & \mbox{$\zeta_{707}$} & \mbox{$100$} & \mbox{$\zeta_{101}$} \\
\hline
\end{tabular}
\caption{Primitive roots of unity jointed to \mbox{$\QQ(\bi)$}.}\label{table1}
\end{table*}
\begin{table*}
\centering
\begin{tabular}{|c|c||c|c||c|c||c|c|}
\hline
  Extension & Roots & Extension & Roots & Extension & Roots & Extension & Roots \\
  Degrees & of Unity & Degrees & of Unity & Degrees & of Unity & Degrees & of Unity \\
\hline
  &  & \mbox{$26$} & \mbox{$\zeta_{53}$} & \mbox{$51$} & \mbox{$\zeta_{409}$} & \mbox{$76$} & \mbox{$\zeta_{761}$} \\
  \mbox{$2$} & \mbox{$\zeta_5$} & \mbox{$27$} & \mbox{$\zeta_{109}$} & \mbox{$52$} & \mbox{$\zeta_{53}$} & \mbox{$77$} & \mbox{$\zeta_{463}$} \\
  \mbox{$3$} & \mbox{$\zeta_7$} & \mbox{$28$} & \mbox{$\zeta_{29}$} & \mbox{$53$} & \mbox{$\zeta_{107}$} & \mbox{$78$} & \mbox{$\zeta_{79}$} \\
  \mbox{$4$} & \mbox{$\zeta_5$} & \mbox{$29$} & \mbox{$\zeta_{59}$} & \mbox{$54$} & \mbox{$\zeta_{163}$} & \mbox{$79$} & \mbox{$\zeta_{317}$} \\
  \mbox{$5$} & \mbox{$\zeta_{11}$} & \mbox{$30$} & \mbox{$\zeta_{31}$} & \mbox{$55$} & \mbox{$\zeta_{253}$} & \mbox{$80$} & \mbox{$\zeta_{187}$} \\
  \mbox{$6$} & \mbox{$\zeta_{7}$} & \mbox{$31$} & \mbox{$\zeta_{311}$} & \mbox{$56$} & \mbox{$\zeta_{113}$} & \mbox{$81$} & \mbox{$\zeta_{163}$} \\
  \mbox{$7$} & \mbox{$\zeta_{29}$} & \mbox{$32$} & \mbox{$\zeta_{128}$} & \mbox{$57$} & \mbox{$\zeta_{361}$} & \mbox{$82$} & \mbox{$\zeta_{415}$} \\
  \mbox{$8$} & \mbox{$\zeta_{17}$} & \mbox{$33$} &$\zeta_{161}$  & \mbox{$58$} & \mbox{$\zeta_{233}$} & \mbox{$83$} & \mbox{$\zeta_{167}$} \\
  \mbox{$9$} & \mbox{$\zeta_{19}$} & \mbox{$34$} & \mbox{$\zeta_{103}$} & \mbox{$59$} & \mbox{$\zeta_{709}$} & \mbox{$84$} & \mbox{$\zeta_{203}$} \\
  \mbox{$10$} & \mbox{$\zeta_{25}$} & \mbox{$35$} & \mbox{$\zeta_{71}$} & \mbox{$60$} & \mbox{$\zeta_{155}$} & \mbox{$85$} & \mbox{$\zeta_{1133}$} \\
  \mbox{$11$} & \mbox{$\zeta_{23}$} & \mbox{$36$} & \mbox{$\zeta_{95}$} & \mbox{$61$} & \mbox{$\zeta_{367}$} & \mbox{$86$} & \mbox{$\zeta_{173}$} \\
  \mbox{$12$} & \mbox{$\zeta_{35}$} & \mbox{$37$} & \mbox{$\zeta_{149}$} & \mbox{$62$} & \mbox{$\zeta_{961}$} & \mbox{$87$} & \mbox{$\zeta_{349}$} \\
  \mbox{$13$} & \mbox{$\zeta_{53}$} & \mbox{$38$} & \mbox{$\zeta_{361}$} & \mbox{$63$} & \mbox{$\zeta_{127}$} & \mbox{$88$} & \mbox{$\zeta_{89}$} \\
  \mbox{$14$} & \mbox{$\zeta_{29}$} & \mbox{$39$} & \mbox{$\zeta_{79}$} & \mbox{$64$} & \mbox{$\zeta_{256}$} & \mbox{$89$} & \mbox{$\zeta_{179}$} \\
  \mbox{$15$} & \mbox{$\zeta_{31}$} & \mbox{$40$} & \mbox{$\zeta_{187}$} & \mbox{$65$} & \mbox{$\zeta_{131}$} & \mbox{$90$} & \mbox{$\zeta_{209}$} \\
  \mbox{$16$} & \mbox{$\zeta_{17}$} & \mbox{$41$} & \mbox{$\zeta_{83}$} & \mbox{$66$} & \mbox{$\zeta_{161}$} & \mbox{$91$} & \mbox{$\zeta_{911}$} \\
  \mbox{$17$} & \mbox{$\zeta_{103}$} & \mbox{$42$} & \mbox{$\zeta_{43}$} & \mbox{$67$} & \mbox{$\zeta_{269}$} & \mbox{$92$} & \mbox{$\zeta_{235}$} \\
  \mbox{$18$} & \mbox{$\zeta_{19}$} & \mbox{$43$} & \mbox{$\zeta_{173}$} & \mbox{$68$} & \mbox{$\zeta_{137}$} & \mbox{$93$} & \mbox{$\zeta_{373}$} \\
  \mbox{$19$} & \mbox{$\zeta_{191}$} & \mbox{$44$} & \mbox{$\zeta_{89}$} & \mbox{$69$} & \mbox{$\zeta_{139}$} & \mbox{$94$} & \mbox{$\zeta_{283}$} \\
  \mbox{$20$} & \mbox{$\zeta_{25}$} & \mbox{$45$} & \mbox{$\zeta_{181}$} & \mbox{$70$} & \mbox{$\zeta_{211}$} & \mbox{$95$} & \mbox{$\zeta_{191}$} \\
  \mbox{$21$} & \mbox{$\zeta_{43}$} & \mbox{$46$} & \mbox{$\zeta_{139}$} & \mbox{$71$} & \mbox{$\zeta_{569}$} & \mbox{$96$} & \mbox{$\zeta_{896}$} \\
  \mbox{$22$} & \mbox{$\zeta_{67}$} & \mbox{$47$} & \mbox{$\zeta_{283}$} & \mbox{$72$} & \mbox{$\zeta_{323}$} & \mbox{$97$} & \mbox{$\zeta_{389}$} \\
  \mbox{$23$} & \mbox{$\zeta_{47}$} & \mbox{$48$} & \mbox{$\zeta_{119}$} & \mbox{$73$} & \mbox{$\zeta_{293}$} & \mbox{$98$} & \mbox{$\zeta_{197}$} \\
  \mbox{$24$} & \mbox{$\zeta_{119}$} & \mbox{$49$} & \mbox{$\zeta_{197}$} & \mbox{$74$} & \mbox{$\zeta_{149}$} & \mbox{$99$} & \mbox{$\zeta_{199}$} \\
  \mbox{$25$} & \mbox{$\zeta_{101}$} & \mbox{$50$} & \mbox{$\zeta_{101}$} & \mbox{$75$} & \mbox{$\zeta_{601}$} & \mbox{$100$} & \mbox{$\zeta_{101}$} \\
\hline
\end{tabular}
\caption{Primitive roots of unity jointed to \mbox{$\QQ(\zeta_3)$}.}\label{table2}
\end{table*}
\begin{table*}
\centering
\begin{tabular}{|c||c|c|c|}\hline
  \multicolumn{4}{|c|}{\mbox{$n=8$}}\\
  \hline Coding Schemes & New Code & \cite{LC09} & \cite{GX09}\\
  \hline &&& \\ \mbox{$\xi(\mathbf{S})$} & \mbox{$\frac{1}{278130^8}$} & \mbox{$\frac{1}{414430^8}$} & \mbox{$\frac{1}{888380^8}$}\\ &&&\\
  \hline \mbox{$\gamma$} & \mbox{$1+\bi$} & \multicolumn{2}{c|}{\mbox{$2+\bi$}}\\
  \hline \mbox{$\eta$} & \multicolumn{2}{c|}{\mbox{$\zeta_{17}+\zeta_{17}^{16}$}} & \mbox{$\zeta_{32}+\zeta_{32}^{15}$}\\
  \mbox{$\sigma(\eta)$} & \multicolumn{2}{c|}{\mbox{$\zeta_{17}^3+\zeta_{17}^{14}$}} & \mbox{$\zeta_{32}^5+\zeta_{32}^{11}$}\\
  \mbox{$\sigma^2(\eta)$} & \multicolumn{2}{c|}{\mbox{$\zeta_{17}^8+\zeta_{17}^{9}$}} & \mbox{$\zeta_{32}^{25}+\zeta_{32}^{23}$}\\
  \mbox{$\sigma^3(\eta)$} & \multicolumn{2}{c|}{\mbox{$\zeta_{17}^7+\zeta_{17}^{10}$}} & \mbox{$\zeta_{32}^{29}+\zeta_{32}^{19}$}\\
  \mbox{$\sigma^4(\eta)$} & \multicolumn{2}{c|}{\mbox{$\zeta_{17}^4+\zeta_{17}^{13}$}} & \mbox{$\zeta_{32}^{17}+\zeta_{32}^{31}$}\\
  \mbox{$\sigma^5(\eta)$} & \multicolumn{2}{c|}{\mbox{$\zeta_{17}^5+\zeta_{17}^{12}$}} & \mbox{$\zeta_{32}^{21}+\zeta_{32}^{27}$}\\
  \mbox{$\sigma^6(\eta)$} & \multicolumn{2}{c|}{\mbox{$\zeta_{17}^2+\zeta_{17}^{15}$}} & \mbox{$\zeta_{32}^{9}+\zeta_{32}^{7}$}\\
  \mbox{$\sigma^7(\eta)$} & \multicolumn{2}{c|}{\mbox{$\zeta_{17}^6+\zeta_{17}^{11}$}} & \mbox{$\zeta_{32}^{13}+\zeta_{32}^{3}$}\\
  \hline \hline \multicolumn{4}{|c|}{\mbox{$n=16$}}\\
  \hline Coding Schemes & New Code & \cite{LC09} & \cite{GX09}\\
  \hline &&& \\ \mbox{$\xi(\mathbf{S})$} & \mbox{$\frac{1}{6016^{16}}$} & \mbox{$\frac{1}{11776^{16}}$} & \mbox{$\frac{1}{\left(1.7051\cdot 10^{11}\right)^{16}}$}\\ &&&\\
  \hline \mbox{$\gamma$} & \mbox{$1+\bi$} & \multicolumn{2}{c|}{\mbox{$2+\bi$}}\\
  \hline \mbox{$\eta$} & \multicolumn{2}{c|}{\mbox{$\zeta_{17}$}} & \mbox{$\zeta_{64}+\zeta_{64}^{31}$}\\
  \mbox{$\sigma(\eta)$} & \multicolumn{2}{c|}{\mbox{$\zeta_{17}^3$}} & \mbox{$\zeta_{64}^5+\zeta_{64}^{27}$}\\
  \mbox{$\sigma^2(\eta)$} & \multicolumn{2}{c|}{\mbox{$\zeta_{17}^9$}} & \mbox{$\zeta_{64}^{25}+\zeta_{64}^{7}$}\\
  \mbox{$\sigma^3(\eta)$} & \multicolumn{2}{c|}{\mbox{$\zeta_{17}^{10}$}} & \mbox{$\zeta_{64}^{61}+\zeta_{64}^{35}$}\\
  \mbox{$\sigma^4(\eta)$} & \multicolumn{2}{c|}{\mbox{$\zeta_{17}^{13}$}} & \mbox{$\zeta_{64}^{49}+\zeta_{64}^{47}$}\\
  \mbox{$\sigma^5(\eta)$} & \multicolumn{2}{c|}{\mbox{$\zeta_{17}^5$}} & \mbox{$\zeta_{64}^{53}+\zeta_{64}^{43}$}\\
  \mbox{$\sigma^6(\eta)$} & \multicolumn{2}{c|}{\mbox{$\zeta_{17}^{15}$}} & \mbox{$\zeta_{64}^{9}+\zeta_{64}^{23}$}\\
  \mbox{$\sigma^7(\eta)$} & \multicolumn{2}{c|}{\mbox{$\zeta_{17}^{11}$}} & \mbox{$\zeta_{64}^{45}+\zeta_{64}^{51}$}\\
  \mbox{$\sigma^8(\eta)$} & \multicolumn{2}{c|}{\mbox{$\zeta_{17}^{16}$}} & \mbox{$\zeta_{64}^{33}+\zeta_{64}^{63}$}\\
  \mbox{$\sigma^9(\eta)$} & \multicolumn{2}{c|}{\mbox{$\zeta_{17}^{14}$}} & \mbox{$\zeta_{64}^{37}+\zeta_{64}^{59}$}\\
  \mbox{$\sigma^{10}(\eta)$} & \multicolumn{2}{c|}{\mbox{$\zeta_{17}^{8}$}} & \mbox{$\zeta_{64}^{57}+\zeta_{64}^{39}$}\\
  \mbox{$\sigma^{11}(\eta)$} & \multicolumn{2}{c|}{\mbox{$\zeta_{17}^{7}$}} & \mbox{$\zeta_{64}^{29}+\zeta_{64}^{3}$}\\
  \mbox{$\sigma^{12}(\eta)$} & \multicolumn{2}{c|}{\mbox{$\zeta_{17}^{4}$}} & \mbox{$\zeta_{64}^{17}+\zeta_{64}^{15}$}\\
  \mbox{$\sigma^{13}(\eta)$} & \multicolumn{2}{c|}{\mbox{$\zeta_{17}^{12}$}} & \mbox{$\zeta_{64}^{21}+\zeta_{64}^{11}$}\\
  \mbox{$\sigma^{14}(\eta)$} & \multicolumn{2}{c|}{\mbox{$\zeta_{17}^{2}$}} & \mbox{$\zeta_{64}^{41}+\zeta_{64}^{55}$}\\
  \mbox{$\sigma^{15}(\eta)$} & \multicolumn{2}{c|}{\mbox{$\zeta_{17}^{6}$}} & \mbox{$\zeta_{64}^{13}+\zeta_{64}^{19}$}\\
\hline
\end{tabular}
\caption{The normalized diversity products and corresponding coding structures.}\label{table3}
\end{table*}
\section{Comparisons With Existing Codes}
\label{numerical}
  Given a cyclic extension \mbox{$\QQ(\eta,\bi)/\QQ(\bi)$} with degree \mbox{$n$}, a generator \mbox{$\sigma$} of \mbox{$\Gal(\QQ(\eta,\bi)/\QQ(\bi))$}, and a non-norm element \mbox{$\gamma$} of \mbox{$\QQ(\eta,\bi)/\QQ(\bi)$}, the generating full-rate space-time block code with nonvanishing determinants can be expressed as~\cite{EKPKL06,GX09,KR05,LC09} $$\mathbf{S}=\left(\begin{array}{ccccc}s_1 & \gamma\sigma(s_n) & \gamma\sigma^2(s_{n-1}) & \cdots & \gamma\sigma^{n-1}(s_2)\\ s_2 & \sigma(s_1) & \gamma\sigma^2(s_n) & \cdots & \gamma\sigma^{n-1}(s_3)\\ s_3 & \sigma(s_2) & \sigma^2(s_1) & \cdots & \gamma\sigma^{n-1}(s_4)\\ \vdots & \vdots & \vdots & \ddots & \vdots\\ s_n & \sigma(s_{n-1}) & \sigma^2(s_{n-2}) & \cdots & \sigma^{n-1}(s_1)\end{array}\right),$$ where $$s_i= \sum_{j=1}^n x_{i,j}\eta^{j-1}$$ for each \mbox{$i=1,\ldots,n$} with \mbox{$(x_{i,j})_{1\le i,j\le n}\in \ZZ(\bi)^{n\times n}$} representing the \mbox{$n^2$} input symbols. When \mbox{$\gamma$} is an integer non-norm element, the minimum determinant of \mbox{$\mathbf{S}$}, $$\delta(\mathbf{S})= \min_{(x_{i,j})_{1\le i,j\le n}\in \ZZ(\bi)^{n\times n}\setminus \mathbf{0}_{n\times n}}|\det(\mathbf{S})|^2,$$ always equals \mbox{$1$}, where \mbox{$\mathbf{0}_{k\times m}$} denotes the \mbox{$k\times m$} matrix all of whose entries are zero. An asymptotic measure of the performance for \mbox{$\mathbf{S}$} over QAM can be determined via the normalized diversity product~\cite{GX09}: $$\xi(\mathbf{S})= E^{-n}=\left(\sum_{i=0}^{n-1} \left\|\left(\begin{array}{cc}\mathbf{I}_{n-i} & \mathbf{0}_{(n-i)\times i}\\ \mathbf{0}_{i\times(n-i)} & \gamma\mathbf{I}_{i} \end{array}\right)\Xi\right\|^2_F\right)^{-n},$$ where \mbox{$\mathbf{I}_{k}$} is the identity matrix of size \mbox{$k\times k$}, $$\Xi=\left(\begin{array}{ccccc}1 & \eta & \eta^2 & \cdots & \eta^{n-1}\\ 1 & \sigma(\eta) & (\sigma(\eta))^2 & \cdots & (\sigma(\eta))^{n-1}\\ 1 & \sigma^2(\eta) & (\sigma^2(\eta))^2 & \cdots & (\sigma^2(\eta))^{n-1}\\ \vdots & \vdots & \vdots & \ddots & \vdots\\ 1 & \sigma^{n-1}(\eta) & (\sigma^{n-1}(\eta))^2 & \cdots & (\sigma^{n-1}(\eta))^{n-1}\end{array}\right),$$ and \mbox{$\|\cdot\|_F$} stands for the Frobenius norm of input matrix. In other words, \mbox{$E$} represents the total energy needed for encoding \mbox{$\mathbf{S}$}. Table~\ref{table3} lists the normalized diversity products and corresponding coding structures for the new codes constructed in Example~\ref{real} and those in~\cite{GX09} and~\cite{LC09} with \mbox{$n=8$} and \mbox{$16$}; the smaller the absolute value of an integer non-norm element \mbox{$\gamma$}, the larger the normalized diversity product \mbox{$\xi(\mathbf{S})$}.
\section{Conclusions}
\label{conc}
  This paper presents a newly constructive method that proves \mbox{$1+\bi$} is an integer non-norm element with the smallest absolute value over QAM for every \mbox{$n\geq 5$}. The proposed designs achieve better normalized diversity products and the optimal diversity-multiplexing gains tradeoff.


\begin{thebibliography}{10}
\bibitem{AR51}
N. C. Ankeny and C. A. Rogers,
\newblock ``A conjecture of Chowla,"
\newblock {\em The Annals of Mathematics}, Second Series, Vol.53, No.3, pp. 541--550, May 1951.
\bibitem{EKPKL06}
P. Elia, K. R. Kumar, S. A. Pawar, P. V. Kumar, and H.-F. Lu,
\newblock ``Explicit, minimum-delay space-time codes achieving the diversity-multiplexing gain tradeoff,"
\newblock {\em IEEE Transactions on Information Theory}, Vol.52, No.9, pp. 3869--3884, September 2006.
\bibitem{FV02}
I. B. Fesenko and S. V. Vostokov,
\newblock{\em Local Fields and Their Extensions}, 2nd edition, American Mathematical Society, 2002.
\bibitem{GX09}
X. Guo and X.-G. Xia,
\newblock ``An elementary condition for non-norm elements,"
\newblock {\em IEEE Transactions on Information Theory}, Vol.55, No.3, pp. 1080--1085, March 2009.
\bibitem{IR90}
K. Ireland and M. Rosen,
\newblock{\em A Classical Introduction to Modern Number Theory}, 2nd edition, Springer-Verlag New York, 1990.
\bibitem{J96}
G. J. Janusz,
\newblock {\em Algebraic Number Fields}, 2nd edition, American Mathematical Society, 1996.
\bibitem{KR05}
T. Kiran and B. S. Rajan,
\newblock ``STBC-schemes with nonvanishing determinant for certain number of transmit antennas,"
\newblock {\em IEEE Transactions on Information Theory}, Vol.51, No.8, pp. 2984--2992, August 2005.
\bibitem{LC09}
H.-C. Li and M.-Y. Chen,
\newblock ``Generally explicit space-time codes with nonvanishing determinants for arbitrary numbers of transmit antennas,"
\newblock {\em IEEE Transactions on Information Theory}, Vol.55, No.2, pp. 557--563, February 2009.
\bibitem{ZT03}
L. Zheng and D. N. C. Tse,
\newblock  ``Diversity and multiplexing: A fundamental tradeoff in multiple-antenna channels,"
\newblock {\em IEEE Transactions on Information Theory}, Vol.49, No.5, pp. 1073--1096, May 2003.
\end{thebibliography}
\end{document}